\documentclass[letterpaper,twocolumn,nofootinbib,aps,superscriptaddress]{revtex4}

\usepackage{amsmath}
\usepackage{amssymb}
\usepackage{hyphenat}
\usepackage{amsthm}
\usepackage{color}
\usepackage{bbm}
\usepackage{bm}
\usepackage{upgreek}
\usepackage[hyperfootnotes=false]{hyperref}
\usepackage{footnotebackref}
\usepackage{graphicx}
\usepackage[caption=false]{subfig}
\usepackage[figure]{hypcap}
\usepackage{pdftexcmds}
\usepackage{mathrsfs}

\DeclareFontFamily{U}{cbgreek}{}
\DeclareFontShape{U}{cbgreek}{m}{n}{
        <-6>    grmn0500
        <6-7>   grmn0600
        <7-8>   grmn0700
        <8-9>   grmn0800
        <9-10>  grmn0900
        <10-12> grmn1000
        <12-17> grmn1200
        <17->   grmn1728
      }{}
\DeclareFontShape{U}{cbgreek}{bx}{n}{
        <-6>    grxn0500
        <6-7>   grxn0600
        <7-8>   grxn0700
        <8-9>   grxn0800
        <9-10>  grxn0900
        <10-12> grxn1000
        <12-17> grxn1200
        <17->   grxn1728
      }{}

\makeatletter
\newcommand{\normalorbold}{%
  \ifnum\pdf@strcmp{\math@version}{bold}=\z@ bx\else m\fi
}
\makeatother

\newtheorem{theorem}{Theorem}
\newtheorem{athm}{Theorem}[section]

\newtheorem{obsm}[theorem]{Observation}
\newtheorem{defm}{Definition}
\newtheorem{obs}[athm]{Observation}

\newtheorem*{note}{Note}

\newtheorem{definition}{Definition}[section]

\newcommand*{\eins}{\ensuremath{\mathbbm 1}}
\def\gbm#1{{\let\phi\upphi \let\lambda\uplambda \let\mu\upmu \let\rho\uprho \let\sigma\upsigma \let\tau\uptau \let\theta\uptheta \let\eta\upeta \bm{#1}}}
\newcommand\scong{\mathrel{\substack{\sim\\[-.2ex]
                      \subseteq}}}

\newcommand\gt{\mathrel{\stackrel{\makebox[0pt]{\mbox{\normalfont\tiny{CVTO}}}}{\longmapsto}}}

\newcommand*{\bbR}{\mathbb{R}}

\newcommand*{\bbN}{\mathbb{N}}

\newcommand*{\cV}{\mathcal{V}}

\newcommand*{\bA}{\mathbf{A}}

\newcommand*{\bS}{\mathbf{S}}

\newcommand*{\ket}[1]{\left|#1\right\rangle}

\newcommand*{\Tr}{\mathrm{Tr}}

\newcommand*{\fr}[2]{\frac{#1}{#2}}

\newcommand{\vect}[1]{\mathbf{#1}}

\newcommand{\be}{\begin{equation}}
\newcommand{\ee}{\end{equation}}
\newcommand{\n}{\textendash}
\newcommand{\m}{\textemdash}

\hypersetup{
    pdftoolbar=true,        
    pdfmenubar=true,        
    pdffitwindow=false,     
    pdfstartview={FitH},    
}

\begin{document}
\title{Thermodynamic resources in continuous\hyp variable quantum systems}
\date{\today}

\author{Varun Narasimhachar}
\email{nvarun@ntu.edu.sg}
\affiliation{
School of Physical and Mathematical Sciences, Nanyang Technological University, 637371 Singapore, Singapore
}
\affiliation{Complexity Institute, Nanyang Technological University, 637335 Singapore, Singapore}

\author{Syed Assad}
\affiliation{
School of Physical and Mathematical Sciences, Nanyang Technological University, 637371 Singapore, Singapore
}
\affiliation{Centre for Quantum Computation and Communication Technology, Department of Quantum Science, The Australian National University, Canberra ACT 2600, Australia}

\author{Felix C.\ Binder}
\affiliation{Institute for Quantum Optics and Quantum Information - IQOQI Vienna, Austrian Academy of Sciences, Boltzmanngasse 3, 1090 Vienna, Austria}

\author{Jayne Thompson}
\affiliation{Centre for Quantum Technologies, National University of Singapore, 3 Science Drive 2, 117543 Singapore, Singapore}

\author{Benjamin Yadin}
\affiliation{School of Mathematical Sciences, University of Nottingham, Nottingham NG7 2NR, United Kingdom}
\affiliation{Atomic and Laser Physics, Clarendon Laboratory, University of Oxford, Parks Road, Oxford, OX1 3PU, UK}

\author{Mile Gu}
\email{mgu@quantumcomplexity.org}
\affiliation{
School of Physical and Mathematical Sciences, Nanyang Technological University, 637371 Singapore, Singapore
}
\affiliation{Complexity Institute, Nanyang Technological University, 637335 Singapore, Singapore}
\affiliation{Centre for Quantum Technologies, National University of Singapore, 3 Science Drive 2, 117543 Singapore, Singapore}

\begin{abstract} Thermodynamic resources, beyond their well\hyp known usefulness in work extraction and other thermodynamic tasks, are often important also in tasks that are not evidently thermodynamic. Here we develop a framework for identifying such resources in diverse applications of bosonic continuous\hyp variable systems. Introducing the class of \emph{bosonic linear thermal operations} to model operationally\hyp feasible processes, we apply this model to identify uniquely quantum properties of bosonic states that refine classical notions of thermodynamic resourcefulness. Among these are (1) a suite of temperature\hyp like quantities generalizing the equilibrium temperature to quantum, non\hyp equilibrium scenarios; (2) signal\hyp to\hyp noise ratios quantifying a system's capacity to carry information in phase\hyp space displacement; and (3) well\hyp established non\hyp classicality measures quantifying the resolution in sensing and parameter estimation tasks.

\end{abstract}

\maketitle

\section{Introduction}
Continuous\hyp variable (CV) quantum systems play an integral role in both the historical development of thermodynamics \cite{Einstein07,Debye12} and the recent surge of quantum technologies\m from ultra\hyp large entangled clusters \cite{cluster1,cluster2} to cryptography and metrology \cite{ZMD06,AAU+09,ARL14,BFH16,KR17,BGBP17,KFIT17,LBBH18,Friis2018,PBGWL18,PAB+19}. Many of these applications rely on non\hyp classical states of CV systems\m for instance, squeezed states, which exhibit quantum fluctuations below the vacuum level in certain quadratures.

The resources underlying operational tasks often turn out either to be fundamentally thermodynamic, or else to have a distinctive thermodynamic aspect at the least. This has motivated many to examine the resource of non\hyp classicality from a thermodynamic perspective. This active research program has already shown that heat engines using squeezed thermal reservoirs perform beyond Carnot efficiency \cite{RAS+14,NGKK16,KFIT17,NMGKK18}. Yet, many questions naturally arise: How do we generalize standard concepts of thermal physics to such quantum states? Is it even meaningful to speak of temperatures in such general, non\hyp equilibrium settings?\m and so on. Such questions motivate us to develop a systematic characterization of the thermodynamic resources contained within bosonic CV systems in general quantum states. To this end, a resource\hyp theoretic treatment, which has catalyzed profound advances in understanding the thermodynamics of discrete\hyp variable quantum systems \cite{Janzing2000,Brandao2013,Gour2015,Ng2018}, could likewise stimulate further developments in CV applications by singling out inherently quantum thermodynamic resources.
\begin{figure}[t]
    \includegraphics[width=\columnwidth]{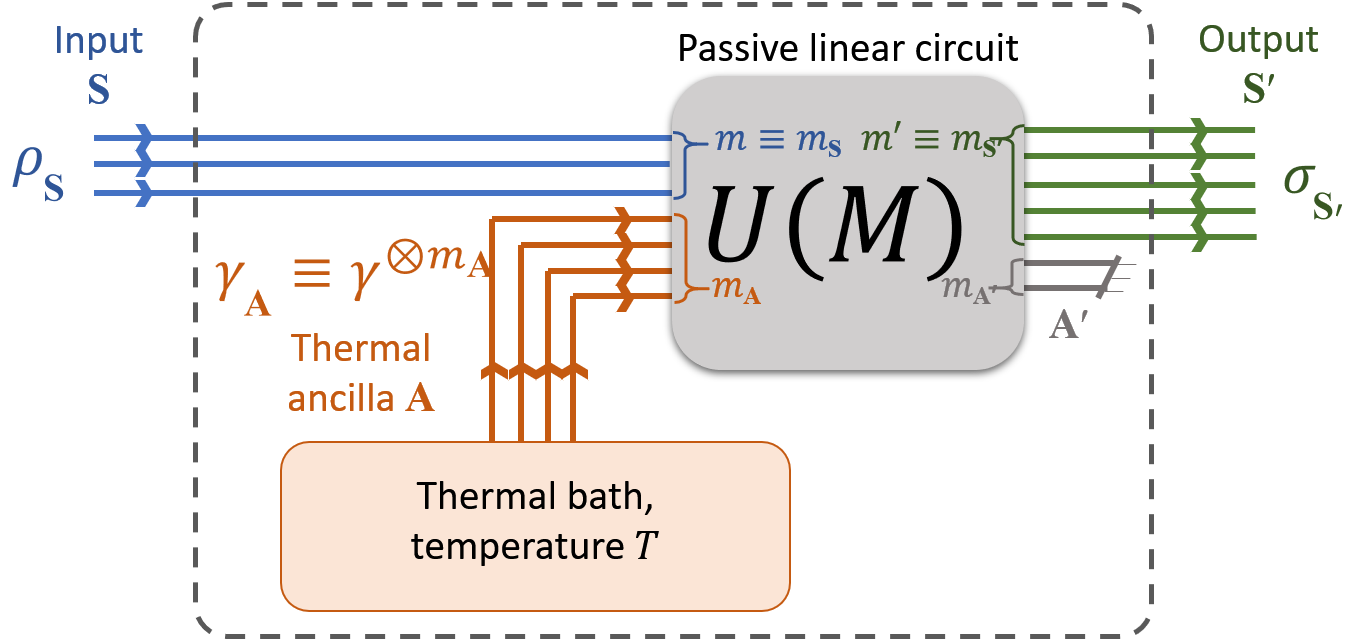}
    \caption{The transformation effected by the contents of the dashed box defines a generic \emph{bosonic linear thermal operation} (BLTO) from an $m$\hyp mode input system $\bS$ to an $m'$\hyp mode output system $\bS'$. The orthogonal symplectic transformation $M$ on the phase space of $\bS\bA$ induces the passive linear unitary $U(M)$ on the corresponding Hilbert space.}\label{figBLTO}
\end{figure}
%


Here, we draw inspiration from this approach and develop an operational framework of quantum thermodynamics for bosonic CV systems subject to Gaussian interactions. We start by defining \emph{bosonic linear thermal operations} (BLTO): the processes that can be enacted on such systems without requiring additional sources of free energy. An operational restriction to BLTO leads to several families of second law\n like statements. Firstly, we identify a spectrum of generalized temperatures for general bosonic states, all of which (1) are sensitive to inherently quantum features of the states; (2) align with the equilibrium notion of temperature for classical states; and, furthermore, (3) equilibrate towards the temperature of the ambient bath under BLTO. Secondly, we illustrate that many known indicators of operational performance and quantifiers of non\hyp classicality\m including phase\hyp space signal\hyp to\hyp noise ratios, squeezing of formation \cite{ILW16}, and phase\hyp space sensing resolution \cite{YBT+18,KTVJ19}\m are non\hyp increasing under BLTO. This thus establishes that many well\hyp known quantifiers of the state's resourcefulness for information\hyp processing and sensing tasks are in fact types of thermodynamic currency.

\section{Framework}\label{secF}
\noindent\textbf{Notation and preliminaries.} While bosonic CV quantum systems occur in many different physical media, it is useful to adopt the terminology of one medium for clarity. Here we will use the language of quantum optics, with the understanding that our results can be readily adapted to other bosonic systems. In this context, an elementary system is a \emph{bosonic mode}, whose local dynamics are governed by a harmonic oscillator Hamiltonian $H = \frac{1}{2} \hbar \omega \left(\hat{p}^2 + \hat{q}^2\right)$. Here $\omega$ denotes a characteristic frequency associated with the mode, and $\left(\hat q,\hat p\right)$ are a conjugate pair of dimensionless quadrature operators, satisfying the canonical commutation relation $\left[\hat q,\hat p\right]=i\hbar$.

In the case of an $m$\hyp mode system, we denote the quadrature operators by $\hat{\vect x}\equiv\left(\hat x_1,\hat x_2\dots,\hat x_{2m}\right)\equiv\left(\hat q_1,\hat p_1,\hat q_2,\hat p_2\dots,\hat q_m,\hat p_m\right)$. For a state whose density operator is $\rho$, we denote the associated first\hyp order quadrature moments $\left\langle\hat{\vect x}\right\rangle_\rho\equiv\left(\left\langle\hat x_1\right\rangle_\rho,\left\langle\hat x_2\right\rangle_\rho\dots,\left\langle\hat x_{2m}\right\rangle_\rho\right)$. The vector $\left\langle\hat{\vect x}\right\rangle_\rho$ lives in a $2m$\hyp dimensional \emph{phase space} $\cV$. The second\hyp order phase\hyp space moments are represented by the covariance matrix $V_\rho$ of $\rho$, defined by
\be
\left(V_\rho\right)_{j,k}:=\fr12\left\langle\left\{\hat x_j-\left\langle\hat x_j\right\rangle_\rho,\hat x_k-\left\langle\hat x_k\right\rangle_\rho\right\}\right\rangle_\rho,
\ee
where $\{\cdot,\cdot\}$ denotes the anti\hyp commutator. We make a choice of units with $\hbar=2$, whereby the covariance matrix of the vacuum state is the identity matrix. The uncertainty constraint on a state's covariance matrix reads $V_\rho+i\Omega_{\cV}\ge0$, where
\be
\Omega_\cV\equiv\Omega_{2m}=\bigoplus_{k=1}^m\left(\begin{array}{cc}0&1\\-1&0\end{array}\right)
\ee
is called the symplectic form on $m$ modes.

In the context of thermodynamics, thermal states of such systems play a central role. We denote by $\gamma$ the density operator of a single mode in the thermal state at ambient temperature. Assuming an ambient temperature $T$, the resulting thermal state is Gaussian with vanishing first moments, and quadrature fluctuations (second moments) given by
\be
 \langle \hat{q}^2\rangle = \langle \hat{p}^2\rangle=\eta:=\coth\left(\fr{\hbar\omega}{k_\mathrm BT}\right).
\ee
When $T=0$, the thermal state coincides with the vacuum state $\ket{0}$, whose uniform quadrature variance $\eta=1$ is called the \emph{vacuum shot noise}. The parameter $\eta$ thus increases monotonically with increasing temperature.\vspace{2ex}

\noindent\textbf{Resource theories.} The resource\hyp theoretic approach to thermodynamics has met with notable success in the past decade \cite{Brandao2013,Gour2015,Ng2018}, extending the second law, Landauer's principle, fluctuation relations, and other thermodynamic cornerstones to the quantum, non\hyp equilibrium regime. In the following, we will borrow some conceptual tools from this approach.

The core idea of a resource theory is to formalize a particular resource (e.g.\ entanglement) operationally in all its complexity, rather than through any single numerical quantifier (e.g.\ the entanglement of formation). This is done by choosing some subset of physical processes that an agent can implement without the ability to generate the given resource; these are formally defined to be the \emph{free operations} of the resource theory (e.g.\ local operations and classical communication, or LOCC). By extension, states that can be prepared by the free operations are called the theory's \emph{free states}. The theory then studies how the free operations may be used to perform useful tasks (possibly by consuming the resource). It also seeks to identify and quantify the resource through \emph{resource monotones}: state functions that are monotonically non\hyp increasing under the free operations.

Note also that the resource\n non\hyp generating property alone does not rigidly determine the class of free operations, leaving some room for variation in the theory. For example, another valid choice of free operations for the resource of entanglement would be separable operations; LOCC, though, are often preferred due to their relevance to operational considerations. We will now discuss the operational considerations that motivate our work.\vspace{2ex}


\noindent\textbf{Bosonic linear thermal operations.} In classical thermodynamics, the second law states that any object in contact with a thermal reservoir at a particular temperature $T$ will drift towards a thermal state of the same temperature\m a process during which it may be possible to extract work. Thus, the key resource is \emph{athermality}, i.e.\ deviation from thermal equilibrium. If a system is, say, at a temperature $T'>T$, the temperature deviation $T'-T$ functions as a natural monotone, allied with the free energy. Our goal here is to generalize these concepts to arbitrary bosonic systems. To this end, we first identify certain elementary processes that (1) do not draw on external free energy, and (2) are operationally inexpensive in state\hyp of\hyp the\hyp art applications.

Consider a system of bosonic modes, with access to an ambient bosonic heat bath at some fixed temperature $T$. As such, preparing auxiliary modes in thermal equilibrium states at this temperature neither requires nor creates free energy: these are the free states in the resource\hyp theoretic sense. No other states can be obtained in this manner\m this includes thermal states at temperatures differing from $T$, but also general, non\hyp thermal states.
%
%

Coupling the system with such auxiliary thermal modes through energy\hyp conserving interactions neither consumes external free energy nor creates any. From a practical standpoint, the interactions that are most feasible on bosonic systems are those that are quadratic in the quadrature operators\m so\hyp called Gaussian operations. Recent work, where arbitrary quadratic local and interaction Hamiltonians are considered, shows that energy\hyp conserving interactions under this constraint effectively decompose into independent processes involving bosonic passive linear interactions between modes of identical frequencies~\cite{SLL+19}. Such interactions correspond to circuits of beam\hyp splitters and phase\hyp shifters in optics. Hence, without loss of generality, we will restrict to interactions of this form between modes of a fixed frequency, which we denote $\omega$. Observing that the thermal noise level $\eta\equiv\coth\left(\hbar\omega/k_\mathrm BT\right)$ has a one\hyp to\hyp one correspondence with the temperature $T$ for fixed $\omega$, this also allows us to use $\eta$ as a placeholder for temperature for mathematical convenience.

In addition to the above operations, we can also freely remove some of the modes from an existing bosonic system. Combining these building blocks, we can formalize the class of free operations as follows:
%
%
%
%
%
\begin{defm}[Bosonic linear thermal operation [BLTO{]}]\label{defBLTO}
Denote the initial system by $\bS$, the number of its constituent modes by $m\equiv m_\bS$, and the thermal noise level corresponding to the ambient temperature by $\eta$. A \emph{bosonic linear thermal operation} (BLTO) is a process realizable through the following steps:
\begin{enumerate}
\item Adding an ancillary system $\bA$ consisting of an arbitrary number $m_\bA$ of elementary modes in uncorrelated thermal states $\gamma_{m_\bA}\equiv\gamma^{\otimes m_\bA}$ with covariance matrix $\eta\eins_{2m_\bA}$.
\item Application of any passive linear unitary on the composite $\bS\bA$.
\item Partial trace over a subsystem $\bA'$ comprising an arbitrary number $m_{\bA'}$ of modes, leaving an output system $\bS'$ of $m'\equiv m_{\bS'}=m+m_{\bA}-m_{\bA'}$ modes.
\end{enumerate}
\end{defm}
Note that the set $\left\{\gamma_k\equiv\gamma^{\otimes k}\right\}_{k\in\bbN}$ of thermal states at the ambient thermal level is closed under BLTO. It is also easy to see that if the modes of the initial system were prepared in thermal states at a level $\eta'\ne\eta$, we could use BLTO to transform them to thermal states intermediate between the two levels, but never to levels outside of this range. This can be interpreted as a semiclassical law of thermalization under BLTO, whereby a system's thermal gradient relative to its ambient bath can never be amplified. We now investigate how such laws can be extended to cases where the initial state of the system is not just a thermal state at some well\hyp defined temperature but, instead, a general quantum state.

\section{Thermodynamic laws under bosonic linear thermal operations}
We will now derive several laws governing the state transitions of modes subject to BLTO evolution. These laws in effect establish resource monotones (cf.\ Section~\ref{secF}) under BLTO. We present the laws in three categories: laws associated with temperature\hyp like quantities; laws concerning the thermal degradation of phase\hyp space displacement considered as a signal carrier; and laws of non\hyp classicality degradation.

\subsection{Thermalization of generalized temperatures}
In equilibrium thermodynamics, a system's temperature determines how it exchanges heat with other systems. In particular, interaction with a heat bath causes the system's temperature to approach that of the bath. This process of thermalization can also be viewed as the gradual dissipation of free energy\m whereby an initial temperature gradient between a system and its environment acts as a resource (commensurate with free energy) that inevitably dissipates, but allows the system to perform useful work in the process.

The degrees of freedom in general non\hyp equilibrium quantum systems, of course, far outnumber those in equilibrium and hence cannot be characterized by a single temperature. A squeezed thermal state, for example, has greater thermal variance (and thus apparent temperature) along one quadrature than another. More dramatically, a two\hyp mode squeezed state can look highly thermal locally on each individual mode, but may in fact have zero global entropy (corresponding to zero temperature). Beyond these, there exist numerous more exotic non\hyp Gaussian states, with no straightforward notion of temperature at all. Nevertheless, do analogous thermalization laws govern such general bosonic states, and if so, what are these laws?

To address this question, we first consider the covariance matrix of a general bosonic state. Recall first that the thermal state has covariance matrix $\eta\eins$, the fixed parameter $\eta$ corresponding to the bath's temperature. In the context of generalized temperatures, we will refer to the value $\eta$ as the \emph{thermal level}. We will consider a value $\eta'>\eta$ to be \emph{super\hyp thermal}, and a value $\eta'<\eta$ to be \emph{sub\hyp thermal}.

Our first class of generalized temperatures is based on the \emph{directional variances} of a state: for a state $\rho$ with covariance matrix $V_\rho$, the directional variance along some unit vector $\vect v$ in the phase space $\cV$ is given by $\vect v^TV_\rho\vect v$. This quantifies the variance in the measurement of a quadrature parallel to $\vect v$.
\begin{defm}[Principal directional temperatures]
For an $m$\hyp mode state $\rho$, we define its \emph{$k^\textnormal{th}$ principal directional temperature} (principal temperature for short) $\tau_k(\rho)$, for $k\in\{1,2\dots,2m\}$, as follows: $\tau_1(\rho)$ is the largest directional variance in the entire phase space; $\tau_2(\rho)$ is the largest directional variance in the subspace orthogonal to a direction associated with $\tau_1(\rho)$, and so on, with each subsequent value defined by maximizing over the subspace remaining after the preceding ones.
\end{defm}
A technical version of the definition is available in Supplementary Note \ref{prtheig}. The principal temperatures are in fact just the $2m$ eigenvalues of the covariance matrix $V_\rho$ of $\rho$, and therefore efficiently computable from $V_\rho$. Experimentally, they can be inferred from the statistics of quadrature measurements. Our first result (proof in Supplementary Note \ref{prtheig}) then states:
\begin{theorem}\label{theigm}
Under bosonic linear thermal operations (BLTO), each of the principal temperatures shifts closer to the thermal level $\eta$, never passing the latter. Specifically, if a BLTO maps $\rho\mapsto\sigma$, then
\begin{enumerate}
\item $\rho$ has no fewer super\hyp thermal principal temperatures than does $\sigma$;
\item $\rho$ has no fewer sub\hyp thermal principal temperatures than does $\sigma$;
\item When arranged in decreasing order, each of $\sigma$'s super\hyp thermal principal temperatures is no higher than the corresponding one of $\rho$;
\item When arranged in increasing order, each of $\sigma$'s sub\hyp thermal principal temperatures is no lower than the corresponding one of $\rho$.
\end{enumerate}
\end{theorem}
\begin{figure}[t]
    \includegraphics[width=\columnwidth]{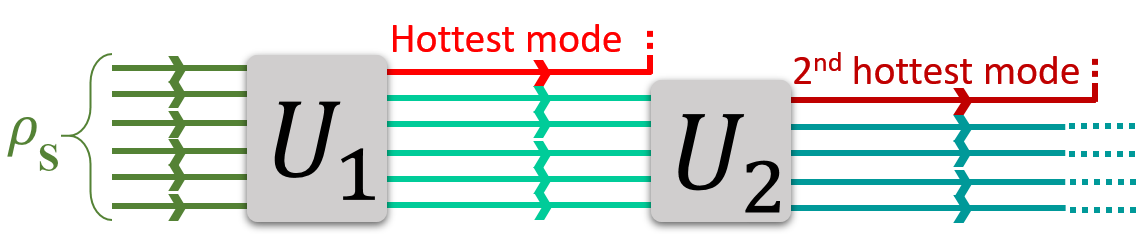}
    \caption{We define the principal mode temperatures as the most extreme effective temperatures in which individual modes can be sequentially isolated using global passive linear operations (labeled by $U_i$ in the illustration).}\label{figMT}
\end{figure}

Thus, each principal directional temperature exhibits behaviour analogous to standard thermalization: when it is super\hyp thermal, interactions with the thermal background will gradually cool it towards equilibrium; when it is sub\hyp thermal, these interactions will heat it towards equilibrium. Provided there is a temperature gradient in any one phase\hyp space direction, the bosonic system overall has some form of free energy\n like resource. Meanwhile, all directional variances of a thermal state are identically thermal (i.e., equal to $\eta$)\m as such, all its principal temperatures align with the conventional definition of temperature in equilibrium thermodynamics.

Given their correspondence to quadrature variances, one is tempted to interpret all principal directional temperatures as apparent temperatures when a bosonic mode is measured in particular phase\hyp space directions. Indeed, in some cases (e.g.\ when considering the temperatures of a squeezed thermal state) this intuition is valid; however, there are exceptions. Consider the case where two thermal modes at different temperatures are coupled through an even beamsplitter, and one of the outgoing modes is then squeezed. The resulting state's principal temperatures correspond to directions in phase space whose simultaneous interpretation as mode quadratures is forbidden by the uncertainty principle. Fortunately, we can define another family of temperature\hyp like measures using a process of ``localized heat distillation'' that does admit a direct physical meaning (technical definition in Supplementary Note \ref{prtheig}):
\begin{defm}[Principal mode temperatures]
For an $m$\hyp mode state $\rho$, we define its \emph{$k^\textnormal{th}$ principal mode temperature} $\mu_k(\rho)$, for $k\in\{1,2\dots,m\}$, as follows: $\mu_1(\rho)$ is the largest (arithmetic) mean principal temperature of a single mode that can be obtained from $\rho$ by closed\hyp system energy\hyp conserving operations; $\mu_2(\rho)$ is the largest single\hyp mode mean principal temperature obtainable from the remaining modes, and so on.
\end{defm}
Fig.~\ref{figMT} schematically illustrates the definition of the principal mode temperatures. This gives the mode temperatures a direct operational meaning as distillable temperatures: given a multi\hyp mode state, what is the hottest single mode that we can distill without drawing on external free energy?\m this defines the principal mode temperature $\mu_1$; once this mode is harnessed, what is the second\hyp hottest mode we can distill?\m call this $\mu_2$; etc. As do the principal directional temperatures, the principal mode temperatures also obey a thermalization law (proof in Supplementary Note \ref{prtheig}):
\begin{theorem}\label{thspeigm}
Under bosonic linear thermal operations (BLTO), each of the principal mode temperatures shifts closer to the thermal level $\eta$, never passing the latter.
\end{theorem}
\begin{figure}[t]
    \includegraphics[width=\columnwidth]{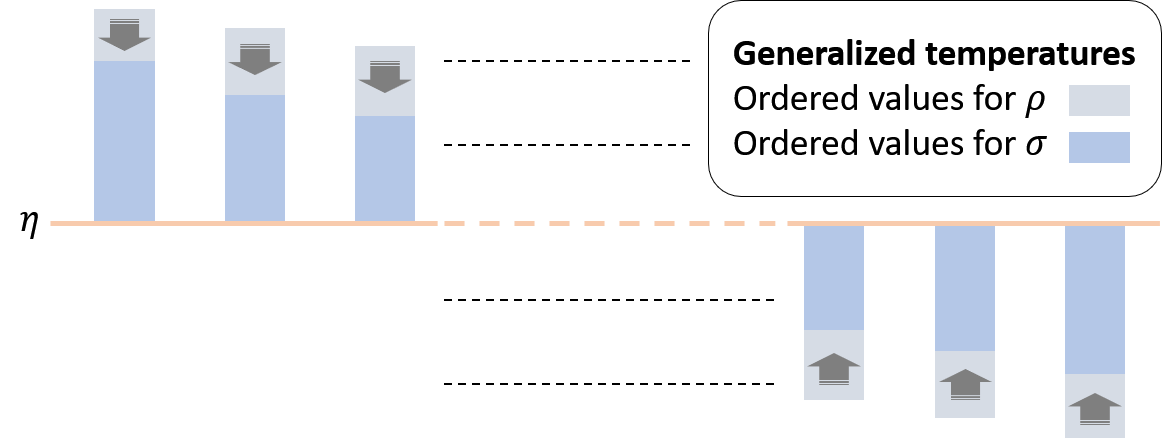}
    \caption{Theorems \ref{theigm} and \ref{thspeigm}: Under a BLTO mapping $\rho\mapsto\sigma$, the principal directional temperatures and mode temperatures of $\sigma$ are each respectively closer to the thermal value $\eta$ than are the corresponding values for $\rho$.}\label{figeig}
\end{figure}

Figure~\ref{figeig} provides a pictorial representation of this thermalization law. Together with their operational interpretation in terms of heat distillation, this law makes the principal mode temperatures a physically meaningful generalization of the equilibrium temperature.



Note that the principal mode temperatures are not the same as the symplectic eigenvalues: the latter correspond to the temperatures of thermal modes \emph{required in preparing} the state (more details below), rather than to any property of modes that can be \emph{extracted from} the state. The symplectic eigenvalues are subject to a somewhat weaker law under BLTO (proof in Supplementary Note \ref{prthsp}):
\begin{theorem}\label{thspm}
Under bosonic linear thermal operations (BLTO), the sub\hyp thermal symplectic eigenvalues cannot shift further away from the thermal level. Specifically, if a BLTO maps $\rho\mapsto\sigma$, then
\begin{enumerate}
\item $\rho$ has no fewer sub\hyp thermal symplectic eigenvalues than does $\sigma$;
\item When arranged in increasing order, each of $\sigma$'s sub\hyp thermal symplectic eigenvalues is no lower than the corresponding one of $\rho$.
\end{enumerate}
\end{theorem}
It is well\hyp known (see, e.g., \cite{WPG+12}) that the symplectic eigenvalues quantify the temperatures of thermal states required in preparing a Gaussian state by Gaussian operations (cf.\ Fig.~\ref{figEBM}). The last theorem then tells us that the sub\hyp thermal symplectic eigenvalues directly quantify the amount of sub\hyp thermal temperature differential required in preparing the state under BLTO. The super\hyp thermal symplectic eigenvalues, on the other hand, are not monotones in that they may sometimes increase under BLTO, albeit not without the initial presence of squeezedness in the state.

\subsection{Signal deterioration laws}
Our next result is a straightforward observation about the first\hyp order phase\hyp space quadrature moments (proof in Supplementary Note \ref{prth1m}):
\begin{obsm}\label{th1m}
If a bosonic linear thermal operation (BLTO) achieves the transformation $\rho\mapsto\sigma$, then
\be
\sum_{k=1}^{2m'}\left|\left\langle\hat x_k\right\rangle_\sigma\right|^2\le\sum_{j=1}^{2m}\left|\left\langle\hat x_j\right\rangle_\rho\right|^2.
\ee
\end{obsm}
\begin{figure}[t]
    \includegraphics[width=\columnwidth]{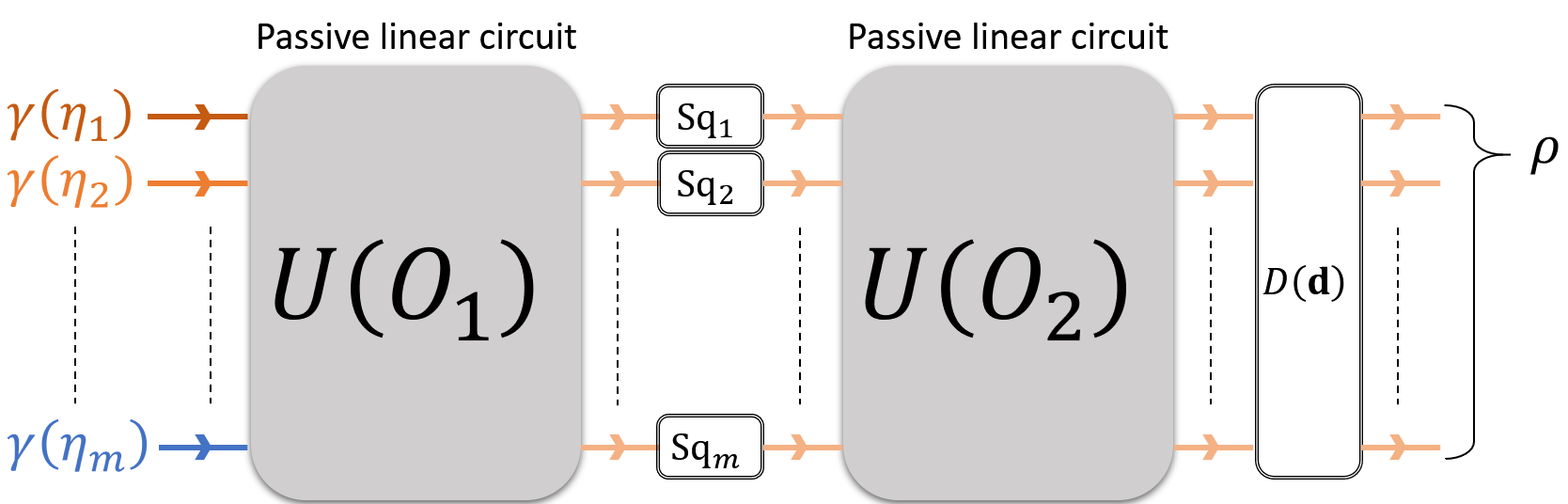}
    \caption{The Williamson and Euler (or Bloch\n Messiah) decompositions allow a generic $m$\hyp mode Gaussian state $\rho$ to be prepared systematically from $m$ uncorrelated single\hyp mode thermal states, represented here by $\gamma\left(\eta_k\right)$. Each Sq${}_k$ is a single\hyp mode squeezing operation, while the $U\left(O_j\right)$ are passive linear unitaries; $D(\vect d)$ is a phase\hyp space displacement operation.}\label{figEBM}
\end{figure}
Thus, if the phase\hyp space displacement in the state is used as a medium to carry information, then the maximum strength of a potential signal deteriorates under BLTO. However, recall Theorem~\ref{theigm}: the super\hyp thermal variances undergo a diminution under BLTO\m possibly counteracting the signal attenuation. Thus, we ask: can the noise reduction possibly compensate for the signal attenuation, resulting in an improvement of the signal\hyp to\hyp noise ratio (SNR)? In order to answer this question, we must formalize the SNR. For an $m$\hyp mode state $\rho$ with first moments $\left\langle\hat{\vect x}\right\rangle_\rho$, the first moment's component along the direction of an arbitrary unit vector $\vect v\in\bbR^{2m}$ in phase space is given by $\vect v^T\left\langle\hat{\vect x}\right\rangle_\rho$. The corresponding directional variance, in terms of the covariance matrix $V_\rho$, is $\vect v^TV_\rho\vect v$ (as discussed before); this can be interpreted as the square of the directional noise. Thus, we can define the \emph{directional SNR} as $\left|\vect v^T\left\langle\hat{\vect x}\right\rangle_\rho/\sqrt{\vect v^TV_\rho\vect v}\right|$. The optimal SNR of $\rho$ then is the maximum directional SNR over the entire phase space. In fact, as with the generalized temperatures, we define an entire family of SNR's (technical definition in Supplementary Note \ref{prsnr}):
\begin{defm}[Principal directional SNR's]\label{defsnr}
For an $m$\hyp mode state $\rho$, we define its \emph{$k^\textnormal{th}$ principal directional signal\hyp to\hyp noise ratio} $\textnormal{SNR}_k\left(\rho\right)$, for $k\in\{1,2\dots,2m\}$, as follows:
$\textnormal{SNR}_1\left(\rho\right)$ is the optimal directional SNR over the entire phase space; $\textnormal{SNR}_2\left(\rho\right)$ is the optimum over the subspace orthogonal to a direction achieving $\textnormal{SNR}_1\left(\rho\right)$, and so on.
\end{defm}
In the same spirit that the principal mode temperatures were defined, we define the following operationally\hyp motivated variants of the principal directional SNR's, restricting the directions to be simultaneously obtainable as quadratures in the phase space (technical definition in Supplementary Note \ref{prsnr}):
\begin{defm}[Principal mode SNR's]\label{defmsnr}
For an $m$\hyp mode state $\rho$, we define its \emph{$k^\textnormal{th}$ principal mode SNR} $\textnormal{MSNR}_k\left(\rho\right)$, for $k\in\{1,2\dots,m\}$, as follows: $\textnormal{MSNR}_1\left(\rho\right)$ is the largest directional SNR in a single mode that can be obtained from $\rho$ by closed\hyp system energy\hyp conserving operations; $\textnormal{MSNR}_2(\rho)$ is the largest directional SNR in a single mode obtainable from the remaining modes, and so on.
\end{defm}
Note that all the principal directional and mode SNR's of a thermal state are zero, by virtue of the first moments' being zero. In general, we have (proof in Supplementary Note \ref{prsnr}):
\begin{theorem}\label{thsnrm}
Under bosonic linear thermal operations (BLTO), the principal directional and mode SNR's can never increase. Specifically, if a BLTO maps $\rho\mapsto\sigma$ with an $m'$\hyp mode output, then
\begin{enumerate}
\item $\textnormal{SNR}_k\left(\sigma\right)\le\textnormal{SNR}_k\left(\rho\right)$ $\forall k\in\{1,2\dots,2m'\}$;
\item $\textnormal{MSNR}_k\left(\sigma\right)\le\textnormal{MSNR}_k\left(\rho\right)$ $\forall k\in\{1,2\dots,m'\}$.
\end{enumerate}
\end{theorem}
Thus, the SNR in every principal component of the phase\hyp space displacement can only deteriorate under BLTO, showing that the signal attenuation effect always trumps any reduction in noise. It is important to note that this result is not of the ``data\hyp processing inequality'' type: that any \emph{specific} information contained in the initial state could only possibly be lost, would be true not only under BLTO but any processing. Rather, Theorem~\ref{thsnrm} is about the usefulness of the displacement degrees of freedom as a \emph{potential} information encoding medium\m \emph{if} these degrees of freedom were used to carry information, then their usefulness for this purpose would deteriorate under BLTO. In particular, if we relaxed BLTO by allowing displacement unitaries, Theorem~\ref{thsnrm} would no longer hold, while of course the data\hyp processing principle would still hold.

\begin{figure*}[t]
  \centering
  \subfloat{
    \includegraphics[width=.3\textwidth]{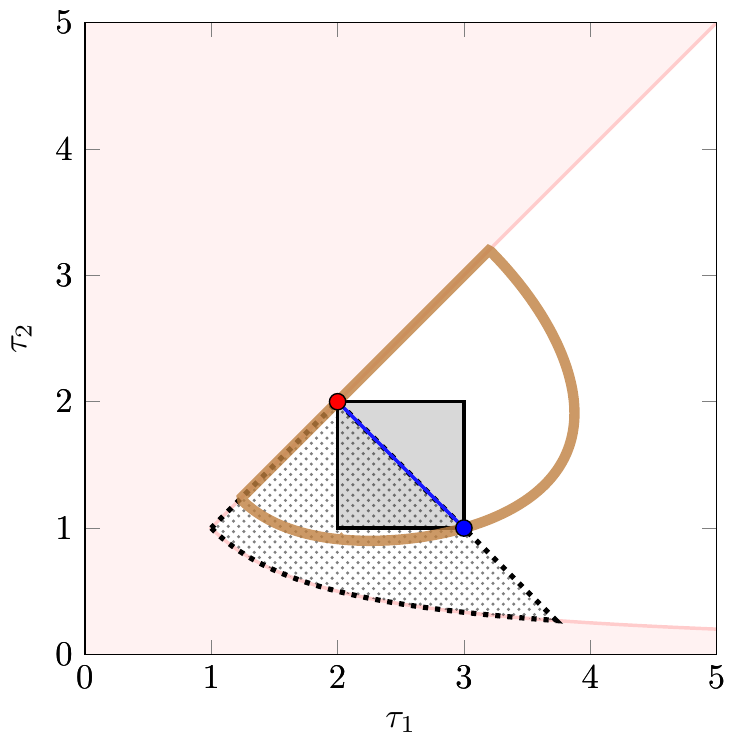}
    }
  \subfloat{
    \includegraphics[width=.3\textwidth]{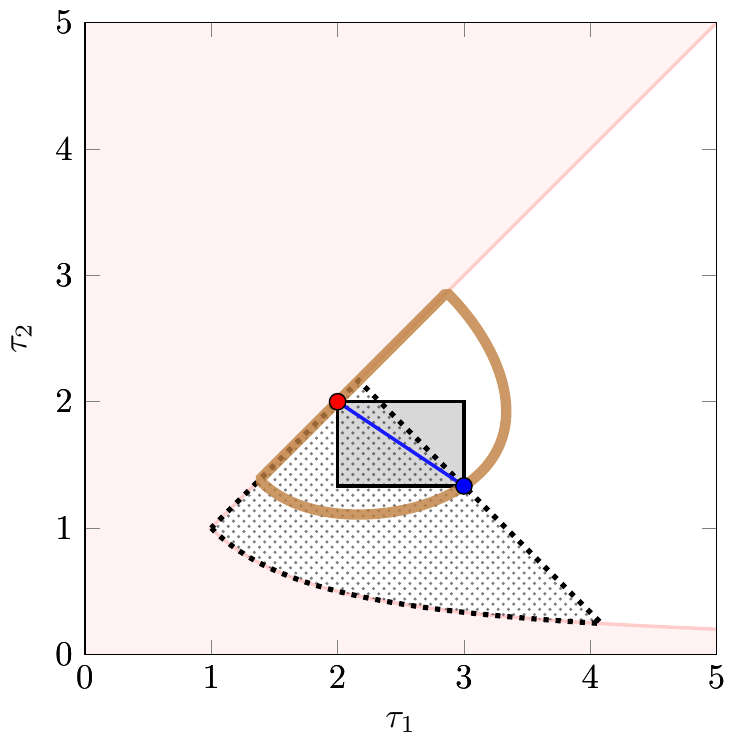}
    }
  \subfloat{
    \includegraphics[width=.3\textwidth]{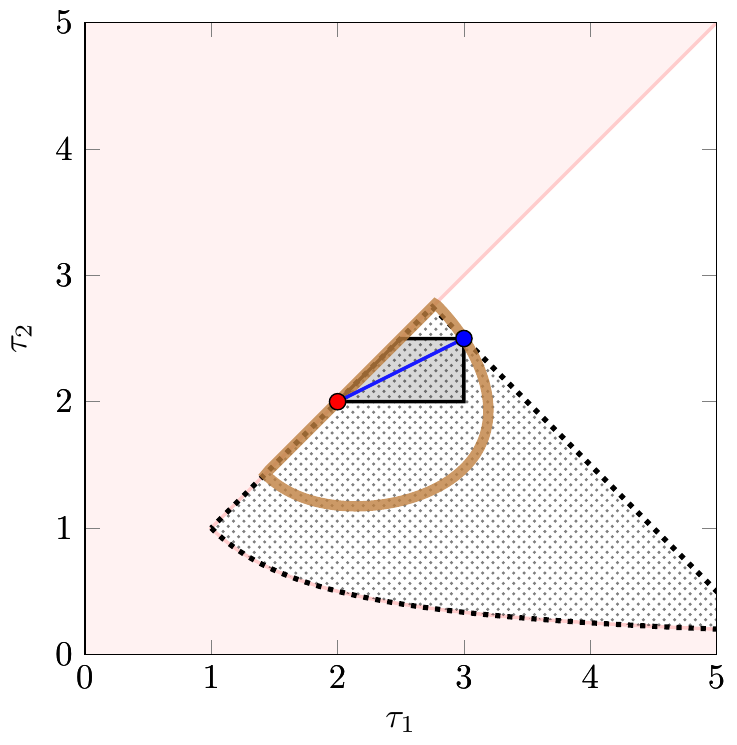}
    }\\
  \subfloat{
    \includegraphics[width=.3\textwidth]{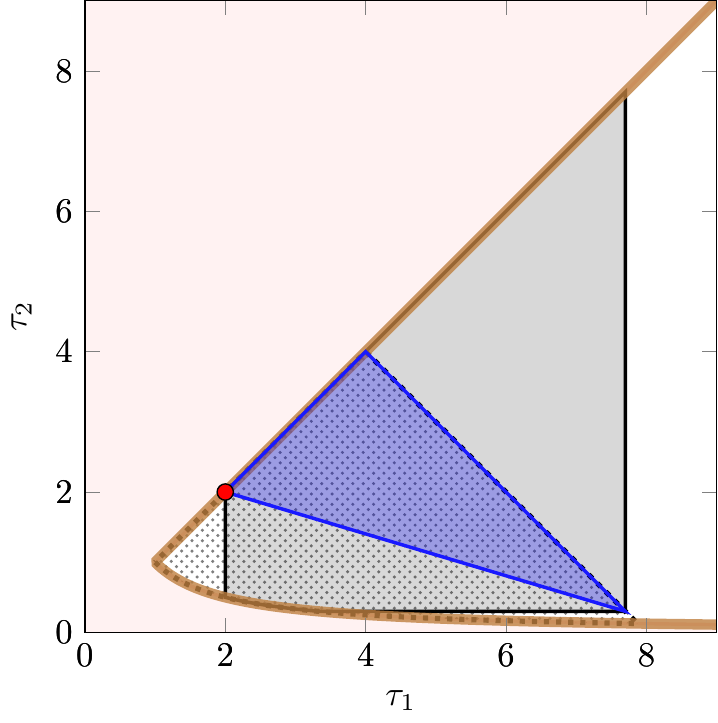}
    }
  \subfloat{
    \includegraphics[width=.3\textwidth]{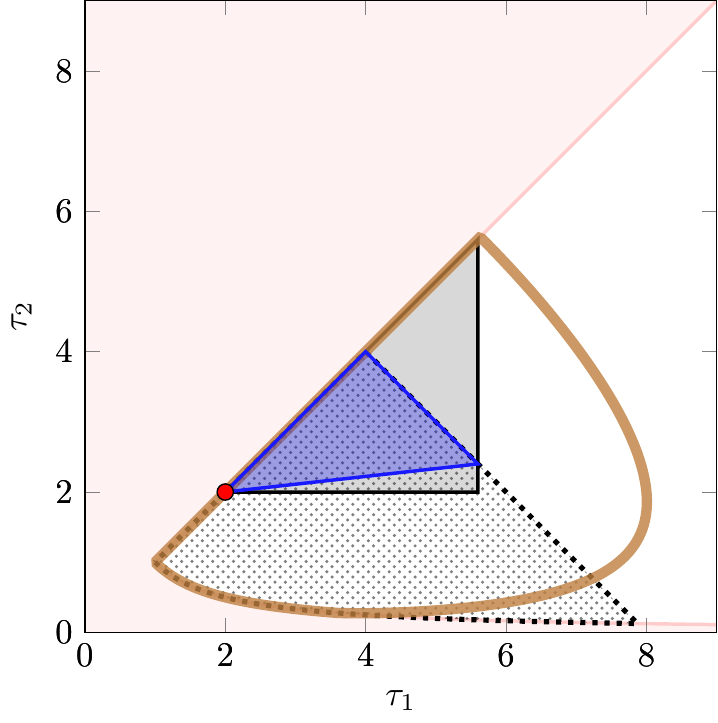}
    }
  \subfloat{
    \includegraphics[width=.3\textwidth]{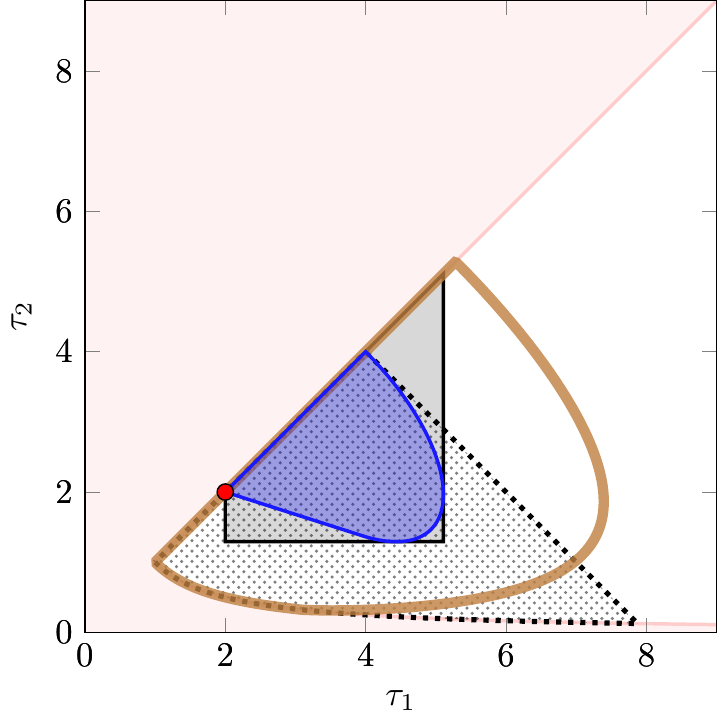}
    }
    \caption{Visualization of some of our thermodynamic laws\m Each example in the top half is associated with a single\hyp mode initial state (marked with a blue dot), while those in the bottom have two\hyp mode initial states (details in Supplementary Note \ref{appex}). The plot region contains potential single\hyp mode states reachable from the given initial state, with the X axis parametrizing $\tau_1$ (the first principal directional temperature), and the Y axis $\tau_2$, of these states. The thermal state is marked with a red dot. The outer pink region marks unphysical states that must therefore be ignored. The blue\hyp shaded region enclosed by solid blue lines depicts all the single\hyp mode states accessible from the given initial state, found using brute\hyp force numerical computation (feasible only in these simple cases): notice that this region shrinks to just a line for single\hyp mode initial states. The grey\hyp shaded region enclosed by solid black contains all final states consistent with Theorem~\ref{theigm}; the dotted region enclosed by the dashed black lines contains those consistent with Theorem~\ref{thspeigm}; finally, the region enclosed by the solid yellow lines contains final states allowed by the monotonicity of the generalized non\hyp equilibrium Helmholtz free energy (i.e.\ relative entropy with respect to the thermal state).}\label{figex1}
\end{figure*}

\subsection{Non\hyp classicality degradation and other inherited laws}
Some notable measures already defined in the literature, and known to have operational significance in other contexts, turn out to be BLTO monotones:
\begin{enumerate}
\item The recently\hyp developed resource theory for CV non\hyp classicality \cite{YBT+18} identified passive linear circuits with classical ancillary systems and measurement\n feed\hyp forward as the class of operations that cannot increase non\hyp classicality as manifested in the negativity of the Glauber\n Sudarshan $P$ function. Since BLTO fall within these operations, all non\hyp classicality measures found in \cite{YBT+18} are also BLTO monotones. These include convex roof extensions of phase\hyp subspace variances, as well as Fisher information\n based measures that quantify the usefulness of a state in the task of detecting phase\hyp space displacement operations. The stronger constraints in BLTO imply that similar Fisher information\hyp based results would hold in connection with the task of detecting a bosonic phase shift.
\item In any resource theory, the distance of a given state from the free states (under any contractive metric) is a monotone. Under BLTO, the thermal states are the only free states. Thus, we can construct numerous monotones of the form $D\left(\rho,\gamma\right)$, where $D$ is contractive. In particular, the ``relative entropy of athermality'', $S(\rho\|\gamma)$, has been identified as a direct analog of the classical Helmholtz free energy for discrete\hyp variable systems \cite{Sec}. This and all similar metric\hyp based measures naturally function as BLTO monotones, provided they have well\hyp defined values.
\item The \emph{squeezing of formation} \cite{ILW16} is defined as the aggregate of the single\hyp mode squeezing required for preparing a given state from unsqueezed resources. This is a BLTO monotone, since BLTO do not allow any squeezing operations or squeezed ancillary modes. Interestingly, it is known \cite{ILW16} that the squeezing of formation can in general be strictly (indeed, unboundedly) smaller than the squeezing resource determined by the canonical Euler (or Bloch\n Messiah) preparation of a Gaussian state (Fig.~\ref{figEBM}), which we may call the \emph{squeezing of unitary formation}. Since BLTO severely restrict the ancillary systems that can be used, it is plausible that the squeezing of unitary formation is also a BLTO monotone; this question remains open.
\end{enumerate}

\section{Illustrative examples}
We now present some illustrations of our results. First, Fig.~\ref{figex1} depicts the application of our results to the problem of determining which states are reachable under BLTO from a given initial state. To simplify the illustration, we consider only the second moments of all states and ignore their other features. The initial states associated with these examples are described in detail in Supplementary Note \ref{appex}. The examples in the top pertain to the case where a single\hyp mode initial state transforms to a single\hyp mode final state, while the ones in the bottom consider a two\hyp mode initial state transforming to a single\hyp mode final state through a BLTO that ends by discarding one mode. In these simple examples, we were able to numerically compute the set of all states reachable by BLTO from a given initial state, with which we then compared the sets of final states compatible with each of our laws. In more complex cases, we expect our results to serve as efficiently computable indicators of state transformation feasibility, which by itself would be computationally cumbersome to determine.

While these examples were chosen arbitrarily to represent a diverse range of cases, we shall now consider a practically relevant special case wherein the initial state is a squeezed thermal state of the same temperature as the bath. In order to motivate this example, consider the semiclassical regime. Here the system's state can deviate from equilibrium with the bath in only one way, namely as a thermal state at temperatures different from the bath's. On the other hand, modes in their full quantum description can contain a fundamentally quantum\hyp mechanical form of athermality: squeezing. Indeed, squeezed thermal states have been investigated as resources to power nano\hyp scale engines at efficiencies surpassing classical bounds \cite{NGKK16,KFIT17,NMGKK18}. By considering squeezed thermal states at the bath temperature, we can study this quantum thermodynamic resource in isolation.
\begin{figure*}[t]
  \centering
  \subfloat{
    \includegraphics[width=.45\textwidth]{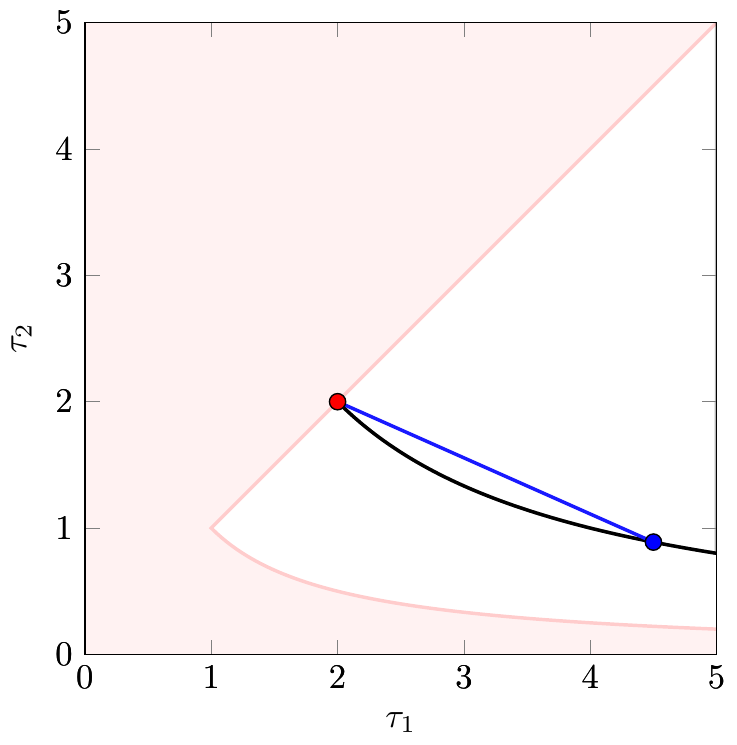}
    }
  \subfloat{
    \includegraphics[width=.45\textwidth]{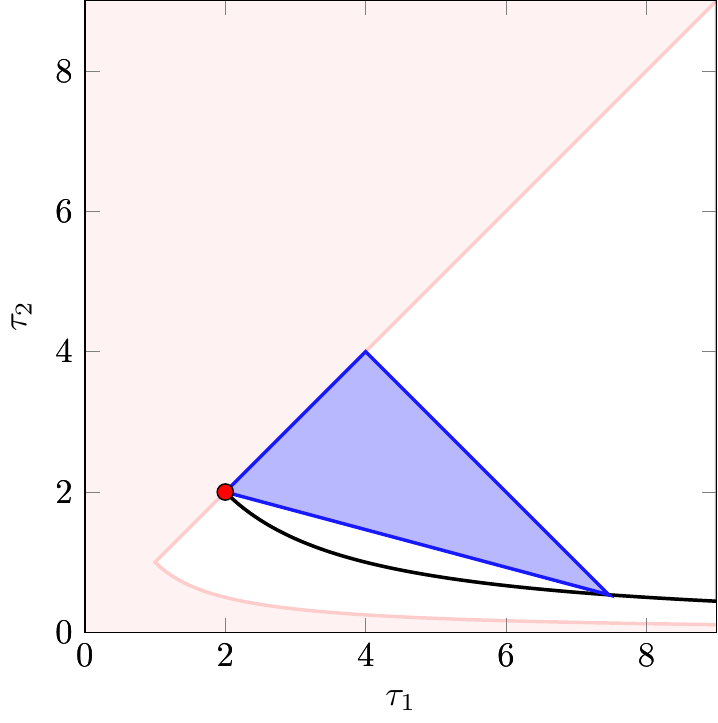}
    }
    \caption{The region enclosed by the solid blue line marks all single\hyp mode states accessible by BLTO starting from a single\hyp mode (left) and a two\hyp mode (right) squeezed thermal state at the same temperature as the bath's. The solid black line shows all squeezed thermal states at the bath temperature. The examples illustrate that genuinely quantum resource in the form of squeezing can be converted to a classical form of resource\m temperature differential relative to the bath.}\label{figex2}
\end{figure*}

Fig.~\ref{figex2} depicts some examples of this category. Evidently, the presence of squeezing in the initial state enables reaching states outside of the solid black set; this can be interpreted as the conversion of the quantum form of athermality, manifested by squeezing, to the classical form of a temperature differential relative to the bath. This interpretation is all the more vivid in the case of the two\hyp mode initial state, where the accessible region contains thermal states at a range of temperatures higher than the bath's\m a purely classical thermodynamic resource. In light of such examples, it is not surprising that squeezed thermal states can be used to overcome classical performance limitations in engines and other applications.

\section{Discussion}
We have built a framework for characterizing those features of a bosonic continuous\hyp variable (CV) quantum system that could constitute thermodynamic resources. In this we have imbibed some of the spirit of recent resource\hyp theoretic treatments of thermodynamics and other topics, framing our results in terms of measures of thermodynamic resourcefulness. Such measures, or resource monotones, are akin to the classical free energies in that they suffer depletion under thermodynamic processes, and also that they quantify a system's usefulness in tasks such as work extraction.

Modelling thermodynamic processes by operationally\hyp motivated bosonic linear thermal operations (BLTO), we have identified a rich spectrum of resource monotones, including quantities that behave like the equilibrium temperature but apply to quantum non\hyp equilibrium states. These quantities acquire immediate operational meaning in terms of phase\hyp space fluctuations, while the other resource monotones are directly related to existing measures of non\hyp classicality or figures of merit for operational tasks in metrology and communication. In applying our framework to two\hyp mode squeezed states, we illustrate that quantum notions of non\hyp classicality (squeezing, entanglement, etc.) can be directly converted to classical notions of free energy (temperature gradients), demonstrating that CV non\hyp classicality has definitive thermodynamic value.

It is worth noting that each of these monotones captures only some aspects of the complex thermodynamic constraints modelled by BLTO. Thus, while each of our laws is a necessary consequence of the constraints of BLTO, the converse is not true: it is possible for non\hyp BLTO processes to be consistent with these laws. Indeed, even considered collectively, our monotones would still only constitute a coarse\hyp grained description of the underlying physical system. This is by design\m thermodynamics has always been an operationalist science, even relative to physics in general. It places operational considerations at its core, using a coarse\hyp grained description of matter to reflect practically viable elements of measurement and control. What these elements are does not need to be set in stone: a thermodynamic framework is most useful when it reflects the operational elements most relevant to the purpose at hand. This can vary both across applications and over time with technological advancement. We have adhered to this operationalist spirit in basing our framework on Gaussian operations, which capture current technological capabilities in bosonic systems.

We hope that our methods can be adapted in the future to build useful thermodynamic theories appropriate to specific applications and reflecting relevant technological constraints. A more complete framework for specialized applications may incorporate non\hyp Gaussian operations, nonlinear interactions such as parametric down\hyp conversion, and hybrid discrete\n continuous system processes such as the Jaynes\n Cummings interaction. The question of what states are freely available opens up another avenue for extension. Indeed, the recently\hyp proposed resource theory of local activity applies to settings where thermal states at all temperatures are available \cite{SJHC19}.

An exciting future direction would be to further understand the operational consequences of our generalized temperatures. One particularly promising avenue is in sensing and metrology. Indeed, closely related notions of non\hyp classicality have already been found to capture the usefulness of states for sensing phase\hyp space displacement \cite{YBT+18,KTVJ19}, while BLTO operations naturally emerge when considering sensing under energetic constraints.

\section*{Acknowledgements}
The authors acknowledge helpful discussions with K.\ C.\ Tan, A.\ Ferraro, D.\ Jennings, P.\ K.\ Lam, C.\ Mukhopadhyay, R.\ Nair, R.\ Takagi, and Y.-D.\ Wu. This research is supported by the National Research Foundation (NRF), Singapore, under its NRFF Fellow programme (Award No.\ NRF-NRFF2016-02), the Lee Kuan Yew Endowment Fund (Postdoctoral Fellowship), Singapore Ministry of Education Tier 1 Grants No.\ MOE2017-T1-002-043, and No.\ FQXi-RFP-1809 from the Foundational Questions Institute and Fetzer Franklin Fund (a donor\hyp  advised fund of Silicon Valley Community Foundation). F.C.B.\ acknowledges funding from the European Union's Horizon 2020 research and innovation programme under the Marie Sk{\l}odowska\hyp Curie grant agreement No.\ 801110 and the Austrian Federal Ministry of Education, Science and Research (BMBWF). B.Y.\ acknowledges financial support from the European Research Council (ERC) under the Starting Grant GQCOP (Grant No.\ 637352). S.A.\ is supported by the Australian Research Council (ARC) under the Centre of Excellence for Quantum Computation and Communication Technology (CE110001027). Any opinions, findings and conclusions or recommendations expressed in this material are those of the author(s) and do not reflect the views of the National Research Foundation, Singapore.


\appendix

\section{Proofs}


\subsection{The form of a generic BLTO}
Towards proving our results, it will help to strip the definition (Def.~\ref{defBLTO}) of a BLTO down into its bare mathematical form using the symplectic geometry of the phase space. Considering the generic BLTO depicted in Fig.~\ref{figBLTO}, denote as before the $m$\hyp mode phase space of the input system $\bS$ by $\cV\cong\mathrm{Sp}\left(2m,\bbR\right)$; let $\cV'\equiv\cV_{\bS'}$ denote the phase space of the output system $\bS'$, and $\cV_\bA$, $\cV_{\bA'}$ those of the ancillary systems. Being a passive linear unitary, $U$ induces on the composite phase space $\cV\oplus\cV_\bA$ a symplectic transformation $M$ that is, besides, orthogonal by virtue of the passivity of $U$. Denoting the phase space quadrature operators of $\bS$ as $\left(\hat{x}_j\right)_{j\in\{1,2\dots,2m\}}\equiv\left(q_1,p_1,q_2,p_2\dots,q_m,p_m\right)$, those of $\bA$ as $\left(\hat{x}_j\right)_{j\in\left\{2m+1,2m+2\dots,2\left(m+m_\bA\right)\right\}}$, those of $\bS'$ as $\left(\hat{x}'_k\right)_{k\in\{1,2\dots,2m'\}}$, and those of $\bA'$ as $\left(\hat{x}'_k\right)_{k\in\left\{2m'+1,2m'+2\dots,2\left(m'+m_{\bA'}\right)\right\}}$ we have
\be
\hat x'_k=\sum_{j=1}^{2\left(m+m_\bA\right)}M_{kj}\hat x_j.
\ee
Noting that the phase\hyp space first moments of thermal modes are identically zero, the resulting transformation of the system first moments looks as follows:
\be
\left\langle\hat x'_k\right\rangle_\sigma=\sum_{j=1}^{2m}M_{kj}\left\langle\hat x_j\right\rangle_\rho.
\ee
Meanwhile, the second moments are encapsulated in the covariance matrix. In order to understand how the latter transforms, we note from the properties of the thermal state that
\be\label{SBLTO1}
V_\sigma=\Pi_{\cV'}M\left(V_\rho\oplus\eta\eins_\bA\right)M^T\Pi_{\cV'},
\ee
where $\Pi_{\cV'}$ is the projector onto the phase space $\cV'$ of $\bS'$. It will be useful for the upcoming proofs to note that the combined operator $\Pi_{\cV'}M$ effects a \emph{symplectic projection}.

\subsection{Proof of Observation~\ref{th1m}}\label{prth1m}
The orthogonality of $M$ implies the conservation of the euclidean norm in phase space:
\be
\sum_k\left|\left\langle\hat x_k\right\rangle_\sigma\right|^2=\sum_j\left|\left\langle\hat x_j\right\rangle_\rho\right|^2.
\ee
Restricting the index $k$ to the output system $\bS'$ immediately yields Observation~\ref{th1m}.\qed

\subsection{Proof of theorems \ref{theigm} and \ref{thspeigm}}\label{prtheig}
We first translate our definitions and theorems to mathematical language; to this end, we start by introducing some notation.
\begin{definition}[Eigenvalues]\label{defeig}
For a symmetric matrix $V$ acting on a (finite $m$)\hyp dimensional real vector space $\cV$, the \emph{$k^{\textnormal{th}}$ largest eigenvalue} of $V$, for $k\in\{1,2\dots,m\}$, is given by
\begin{equation}
    \lambda_k\left[V\right]:=\max_{\cV_k\subseteq\cV}\min_{\vect v\in\cV_k\setminus0}\fr{\vect v^TV\vect v}{\vect v^T\vect v},
\ee
where $\cV_k$ varies over all $k$\hyp dimensional subspaces of $\cV$.
\end{definition}
\begin{definition}\label{defspeig}
For a symmetric $V$ acting on a real, (finite $2m$)\hyp dimensional symplectic vector space $\left(\cV,\Omega_\cV\right)$, define for $k\in\{1,2\dots,m\}$
\begin{equation}\label{speig}
    \nu_k\left[V\right]:=\fr12\max_{\cV_{2k}\scong\cV}\min_{\cV_2\scong\cV_{2k}}\Tr\left[\Pi_{\cV_2}V\right],
\ee
where $\cV_{2k}$ varies over all $2k$\hyp dimensional \emph{symplectic subspaces} of $\cV$, and $\cV_2$ over all 2\hyp dimensional symplectic subspaces of each $\cV_{2k}$.
\end{definition}
Note that the $\nu_k$ are \emph{not} the symplectic eigenvalues of $V$. However, they can be expressed as the eigenvalues of an operator, following the line of argument used in Ref. \cite{YBT+18}, Appendix D:
\begin{obs}\label{nueig}
For any given $V$, define $W:=\fr12\left(V+\Omega V\Omega^T\right)$. Then,
\be
\nu_k\left[V\right]=\lambda_{2k}\left[W\right].
\ee
\end{obs}
\begin{proof}
First, note that
\be
\Tr\left[\Pi_{\cV_2}V\right]=\vect q^TV\vect q+\vect p^TV\vect p,
\ee
where $\vect q$ is an arbitrary unit vector in $\cV_2$ and $\vect p=\Omega_{\cV_2}^T\vect q$ is the quadrature conjugate to $\vect q$. Thus,
\be
\Tr\left[\Pi_{\cV_2}V\right]=\vect q^T\left(V+\Omega V\Omega^T\right)\vect q=2\vect q^TW\vect q.
\ee

$W$ has a special structure in terms of $2 \times 2$ blocks:
\begin{equation}
	W = \begin{pmatrix}
		W^{1,1} & W^{1,2} & \dots \\
		W^{2,1} & W^{2,2} & \dots \\
		\vdots & \vdots & \ddots
		\end{pmatrix}, \quad
		W^{i,j} = \begin{pmatrix}
		W^{i,j}_R & -W^{i,j}_I \\
		W^{i,j}_I & W^{i,j}_R
	\end{pmatrix},
\end{equation}
with the diagonal blocks satisfying $W^{i,i}_I = 0$. This makes the expression for $\nu_k\left[V\right]$ amenable to an isomorphism \cite{dG13} onto a complex vector space of half the dimension: We form $\tilde W \in \mathbb{C}^{m\times m}$ with elements $\tilde{W}_{ij} := W^{i,j}_R + i W^{i,j}_I$, and similarly a vector $\vect r = (r_{1,x},r_{1,p},r_{2,x},r_{2,p},\dots) \in \cV$ is mapped to $\tilde{\vect r} = (r_{1,x} + i r_{1,p}, r_{2,x} + i r_{2,p},\dots) \in\tilde\cV\cong \mathbb{C}^m$. Then $\tilde{\vect r}^\dagger \tilde{W} \tilde{\vect r} = \vect{r}^T W \vect{r}$; in addition, an orthogonal basis in $\mathbb{C}^m$ corresponds to a symplectic basis in $\cV$. Therefore,
\begin{align}
	 \nu_k\left[V\right]&=\max_{\cV_{2k}\scong\cV}\min_{\vect q\in\cV_{2k}}\vect q^TW\vect q\nonumber\\
	 &=\max_{\tilde\cV_{k}\subseteq\cV}\min_{\tilde{\vect q}\in\tilde\cV_{k}}\tilde{\vect q}^T\tilde W\tilde{\vect q}\nonumber\\
	 &=\lambda_k\left[\tilde W\right].
\end{align}
That these are the doubly degenerate eigenvalues of $W$ is seen by inverting the isomorphism to map from the diagonalized form of $\tilde W$ back to the real $2m$\hyp dimensional matrix $\mathrm{diag}\left(\lambda_1\left[\tilde W\right],\lambda_1\left[\tilde W\right],\lambda_2\left[\tilde W\right],\lambda_2\left[\tilde W\right]\dots\right)$.
\end{proof}
\begin{obs}\label{eigpm}
$\lambda_j[W]\ge\lambda_k[W]$ and $\nu_j[W]\ge\nu_k[W]$ whenever $j<k$.
\end{obs}
\begin{obs}
If $\dim(\cV)=2m$, then $\lambda_k[-W]=-\lambda_{2m+1-k}[W]$ and $\nu_k[-W]=-\nu_{m+1-k}[W]$ for all applicable $k$.

It is straightforward to see why this holds for the $\lambda$'s, considering that they are the eigenvalues of a Hermitian operator in a finite\hyp dimensional vector space. It also holds for the $\nu$'s, since by virtue of Observation~\ref{nueig} they, too, are the eigenvalues of a Hermitian operator.
\end{obs}
\begin{note}
In the remainder, any expression with $\pm$ and/or $\mp$ signs is to be interpreted as a \emph{conjunction} of \emph{exactly two} sub\hyp expressions: the one obtained by consistently applying the top sign throughout, and the other by consistently applying the bottom one. The scope of every such consistent application will be clear from the context.
\end{note}
\begin{definition}[Principal directional temperatures]
For an $m$\hyp mode state $\rho$ with covariance matrix $V_\rho$, we define its \emph{$k^\textnormal{th}$ largest principal directional temperature} (principal temperature for short) $\tau_k(\rho)$, for $k\in\{1,2\dots,2m\}$, as
\be
\tau_k\left(\rho\right)\equiv\tau_k^\downarrow:=\lambda_k\left[V_\rho\right].
\ee
\end{definition}
\begin{definition}[Principal mode temperatures]\label{defspeig}
For an $m$\hyp mode state $\rho$ with covariance matrix $V_\rho$, we define its \emph{$k^\textnormal{th}$ principal mode temperature} (mode temperature for short) $\mu_k(\rho)$, for $k\in\{1,2\dots,m\}$, as
\be
\mu_k\left(\rho\right)\equiv\mu_k^\downarrow:=\nu_k\left[V_\rho\right].
\ee
\end{definition}
\begin{obs}
The principal directional and mode temperatures as defined above are arranged in non\hyp increasing order. It follows from Observation~\ref{eigpm} that the same collections of values, arranged in \emph{non\hyp decreasing} order, are given respectively by
\begin{align}
\tau_k^\uparrow:=-\lambda_k\left[-V_\rho\right],\\
\mu_k^\uparrow:=-\nu_k\left[-V_\rho\right].
\end{align}
\end{obs}
Based on the above observations, we now reproduce theorems \ref{theigm} and \ref{thspeigm} of the main text formally in terms of the $\lambda$'s and $\nu$'s:
\begin{athm}[Theorems \ref{theigm} and \ref{thspeigm} of main text]\label{theig}
For a given $m$\hyp mode state $\rho$ and $m'$\hyp mode state $\sigma$ (Fig.~\ref{figBLTO}), denote the corresponding covariance matrices as $\left(V_\rho,V_\sigma\right)$, and define
\begin{align}
k^\pm_\rho&:=\left|\left\{k:\lambda_k\left[\pm V_\rho\right]>\pm\eta\right\}\right|;\nonumber\\
k^\pm_\sigma&:=\left|\left\{k:\lambda_k\left[\pm V_\sigma\right]>\pm\eta\right\}\right|;\nonumber\\
k^{\mathrm{Sp}\pm}_\rho&:=\left|\left\{k:\nu_k\left[\pm V_\rho\right]>\pm\eta\right\}\right|;\nonumber\\
k^{\mathrm{Sp}\pm}_\sigma&:=\left|\left\{k:\nu_k\left[\pm V_\sigma\right]>\pm\eta\right\}\right|.
\end{align}
Then, $\rho\gt\sigma$ only if
\begin{enumerate}
\item $k^\pm_\rho\ge k^\pm_\sigma$ and $k^{\mathrm{Sp}\pm}_\rho\ge k^{\mathrm{Sp}\pm}_\sigma$; and, furthermore,\label{cond1}
\item $\lambda_k\left[\pm V_\rho\right]\ge\lambda_k\left[\pm V_\sigma\right]$ for all $k\le k^\pm_\sigma$, and $\nu_k\left[\pm V_\rho\right]\ge\nu_k\left[\pm V_\sigma\right]$ for all $k\le k^{\mathrm{Sp}\pm}_\sigma$.\label{cond2}
\end{enumerate}
\end{athm}

\begin{proof}
We will go through the proof for the $\nu$'s, which require relatively more careful treatment; we omit the proof for the $\lambda$'s, which proceeds on similar lines but more straightforwardly. Recall that $V_\rho$ and $V_\sigma$ are symmetric positive\hyp semidefinite matrices acting on the respective phase spaces of $\bS$ and $\bS'$, viz.\ $\left(\cV,\Omega\right)\equiv\left(\cV_\bS,\Omega_\bS\right)\cong\left(\bbR^{2m},\Omega_{2m}\right)$ and $\left(\cV',\Omega'\right)\equiv\left(\cV_{\bS'},\Omega_{\bS'}\right)\cong\left(\bbR^{2m'},\Omega_{2m'}\right)$ respectively. Eq.~\eqref{SBLTO1} tells us that $\rho\gt\sigma$ only if there is an orthogonal, symplectic $M$ (acting globally on the symplectic space $\cV\oplus\cV_{\bA}$, where $\bA$ is an ancilla consisting of an arbitrary number $m_\bA\in\bbN$ of modes) such that
\begin{equation}\label{GTW}
    V_\sigma=\Pi_{\cV'}M\left(V_\rho\oplus\eta\eins_\bA\right)M^T\Pi_{\cV'},
\end{equation}
where $\Pi_{\cV'}$ effects an orthogonal projection onto the phase space $\cV'$ of $\bS'$, a symplectic subspace of $\cV\oplus\cV_{\bA}$. Now, for $1\le k\le m'$,
\begin{widetext}
\begin{align}
    \nu_k\left[\pm V_\sigma\right]:=&\fr12\max_{\cV_{2k}\scong\cV'}\min_{\cV_2\scong\cV_{2k}}\Tr\left[\pm\Pi_{\cV_2}V_\sigma\right]\nonumber\\
    =&\fr12\max_{\cV_{2k}\scong\cV'}\min_{\cV_2\scong\cV_{2k}}\Tr\left[\pm\Pi_{\cV_2}\Pi_{\cV'}M\left(V_\rho\oplus\eta\eins_\bA\right)M^T\Pi_{\cV'}\right]\nonumber\\
    \le&\fr12\max_{\cV_{2k}\scong\cV\oplus\cV_{\bA}}\min_{\cV_2\scong\cV_{2k}}\Tr\left[\pm\Pi_{\cV_2}\left(V_\rho\oplus\eta\eins_\bA\right)\right]=\nu_k\left[\pm V_\rho\oplus\eta\eins_\bA\right].\label{VsWr}
\end{align}
\end{widetext}
The second line follows from \eqref{GTW}, and the last line from the fact that the maximization therein subsumes the cases covered by that in the line before. We will now prove that the inequalities \eqref{VsWr} for $1\le k\le m'$ are collectively equivalent to the conjunction of (the symplectic parts of) conditions \ref{cond1} and \ref{cond2} in the statement of Theorem \ref{theig}.

We shall first prove that the former implies the latter. Firstly, it follows from the definition of $k^{\mathrm{Sp}\pm}_\sigma$ that, for $1\le k\le k^{\mathrm{Sp}\pm}_\sigma$,
\be
\nu_k\left[\pm V_\sigma\right]>\pm\eta.
\ee
Meanwhile, for $k>k^{\mathrm{Sp}\pm}_\rho$,
\be
\nu_k\left[\pm V_\rho\oplus\eta\eins_\bA\right]\le\pm\eta.
\ee
This necessitates $k^{\mathrm{Sp}\pm}_\rho\ge k^{\mathrm{Sp}\pm}_\sigma$, i.e.\ condition \ref{cond1}. Provided this holds, we have for $k\le k^{\mathrm{Sp}\pm}_\sigma$ that
\be
\nu_k\left[\pm V_\rho\oplus\eta\eins_\bA\right]=\nu_k\left[\pm V_\rho\right].
\ee
This establishes that inequality~\eqref{VsWr} for $1\le k\le m'$ implies conditions \ref{cond1} and \ref{cond2}. For the converse, suppose \ref{cond1} and \ref{cond2} hold. For $k\le k^{\mathrm{Sp}\pm}_\rho$,
\be
\nu_k\left[\pm V_\rho\oplus\eta\eins_\bA\right]=\nu_k\left[\pm V_\rho\right]>\pm\eta.
\ee
by the definition of $k^{\mathrm{Sp}\pm}_\rho$, thus securing \eqref{VsWr} by virtue of condition~\ref{cond2}. On the other hand, for $k>k^{\mathrm{Sp}\pm}_\rho$, condition~\ref{cond1} implies that $k>k^{\mathrm{Sp}\pm}_\sigma$, so that
\be
\nu_k\left[\pm V_\sigma\right]\le\pm\eta.
\ee
For \eqref{VsWr} to hold, we require this quantity to be bounded above by $\nu_k\left[\pm V_\rho\oplus\eta\eins_\bA\right]$ for some $\bA$ consisting of an arbitrary number of modes. We can achieve this by making the dimensionality of the phase space of $\bA$ larger than $2\left(m'-k^{\mathrm{Sp}\pm}_\rho\right)$, so that for $m'\ge k>k^{\mathrm{Sp}\pm}_\rho$, $\nu_k\left[\pm V_\rho\oplus\eta\eins_\bA\right]=\pm\eta$.
\end{proof}
\subsection{Proof of Theorem~\ref{thspm}}\label{prthsp}
Recall Eq.~\eqref{SBLTO1} relating the input and output covariance matrix under a BLTO:
\be\label{SBLTO2}
V_\sigma=\Pi_{\cV'}M\left(V_\rho\oplus\eta\eins_\bA\right)M^T\Pi_{\cV'},
\ee
where $M$ is some orthogonal symplectic matrix. Let $\tilde V:=M\left(V_\rho\oplus\eta\eins_\bA\right)M^T$. Since $M$ is symplectic, the symplectic spectrum of $\tilde V$ is identical to that of $V_\rho\oplus\eta\eins_\bA$. Let $\eta_1\left[V_\rho\right]\le\eta_2\left[V_\rho\right]\dots\le\eta_m\left[V_\rho\right]$ denote the symplectic eigenvalues of $V_\rho$ in \emph{non\hyp decreasing} order. Define
\be
k_\rho:=\left|\left\{j:\eta_j\left[V_\rho\right]<\eta\right\}\right|,
\ee
i.e., the number of sub\hyp thermal symplectic eigenvalues of $V_\rho$. The symplectic spectrum of $V_\rho\oplus\eta\eins_\bA$\m and, therefore, that of $\tilde V$\m is then given by
\begin{widetext}
\be
\left(\eta_1\left[\tilde V\right],\eta_2\left[\tilde V\right]\dots,\eta_{m+m_\bA}\left[\tilde V\right]\right)=\left(\eta_1\left[V_\rho\right],\eta_2\left[V_\rho\right]\dots,\eta_{k_\rho}\left[V_\rho\right],\eta,\eta\dots,\eta,\eta_{k_\rho+1}\left[V_\rho\right]\dots,\eta_m\left[V_\rho\right]\right),
\ee
\end{widetext}
with $\eta$ appearing $m_\bA$ times on the RHS. Since $V_\sigma$ is obtained from $\tilde V$ by simply removing all rows and columns other than those associated with $\bS'$, the symplectic eigenvalues of $V_\sigma$ and those of $\tilde V$ are related by the interlacing condition \cite{BJ15}
\be
\eta_j\left[V_\sigma\right]\ge\eta_j\left[\tilde V\right].
\ee
But for $j\le k_\rho$, $\eta_j\left[\tilde V\right]=\eta_j\left[V_\rho\right]$. Theorem~\ref{thspm} follows.\qed

\subsection{Proof of Theorem~\ref{thsnrm}}\label{prsnr}
Once again, a mathematical translation of definitions \ref{defsnr} and \ref{defmsnr} will help us prove this theorem.
\begin{definition}[definitions \ref{defsnr} and \ref{defmsnr} of main text]
For an $m$\hyp mode state $\rho$ with covariance matrix $V_\rho$, we define its \emph{$k^\textnormal{th}$ largest principal directional SNR} for $k\in\{1,2\dots,2m\}$ as
\be
\textnormal{SNR}_k\left(\rho\right):=\sqrt{\min_{\cV_\ell}\max_{\vect v\in\cV_\ell\setminus0}\fr{\vect v^T\left\langle\hat{\vect x}\right\rangle_\rho\left\langle\hat{\vect x}\right\rangle_\rho^T\vect v}{\vect v^TV_\rho\vect v}},
\ee
and its \emph{$k^\textnormal{th}$ largest principal mode SNR} for $k\in\{1,2\dots,m\}$ as
\be
\textnormal{MSNR}_k(\rho)=\sqrt{\min_{\cV_{2\ell}\scong\cV}\max_{\vect v\in\cV_{2\ell}\setminus0}\fr{\vect v^T\left\langle\hat{\vect x}\right\rangle_\rho\left\langle\hat{\vect x}\right\rangle_\rho^T\vect v}{\vect v^TV_\rho\vect v}}.
\ee
\end{definition}
Note that $\textbf{SNR}\left(\rho\otimes\gamma_\bA\right)=\textbf{SNR}\left(\rho\right)\oplus0_\bA$ and likewise for the MSNR's. The proof of theorem~\ref{thsnrm} then follows in a straightforward manner along the same lines as the previous proof.

Here we find it opportune to note a simple way to compute the principal directional SNR's:
\begin{obs}\label{Rval}
For an $m$\hyp mode state $\rho$ with first moments and covariance matrix given by $\left(\left\langle\hat{\vect x}\right\rangle_\rho,V_\rho\right)$, define
\be\label{defR}
R_\rho:=V_\rho^{-1/2}\left\langle\hat{\vect x}\right\rangle_\rho\left\langle\hat{\vect x}\right\rangle_\rho^TV_\rho^{-1/2}.
\ee
Then, for $k\in\left\{1,2\dots,2m\right\}$,
\be\label{eqeig}
\textnormal{SNR}_k(\rho)=\sqrt{\lambda_k\left[R_\rho\right]}.
\ee
\end{obs}
\begin{proof}
Let us consider the LHS, with the shorthand $\ell:=2m-k+1$:
\begin{align}
\left|\textnormal{SNR}_k(\rho)\right|^2=&\min_{\cV_\ell}\max_{\vect v\in\cV_\ell\setminus0}\fr{\vect v^T\left\langle\hat{\vect x}\right\rangle_\rho\left\langle\hat{\vect x}\right\rangle_\rho^T\vect v}{\vect v^TV_\rho\vect v}\nonumber\\
=&\min_{\Pi:\Pi^T\Pi=\eins_\ell}\max_{\vect v\in\cV:\Pi\vect v\ne0}\fr{\vect v^T\Pi^T\left\langle\hat{\vect x}\right\rangle_\rho\left\langle\hat{\vect x}\right\rangle_\rho^T\Pi\vect v}{\vect v^T\Pi^T V_\rho\Pi\vect v}.
\end{align}
The strict positive\hyp definiteness of $V_\rho$ (by the uncertainty principle) ensures that $\vect u\equiv V_\rho^{1/2}\Pi\vect v$ is nonzero whenever $\Pi\vect v$ is, and vice\hyp versa; it also ensures that for every $\ell$\hyp dimensional subspace $\cV_\ell$ of $\cV$, there exists a $\Pi$ such that $\mathrm{span}\left(V_\rho^{1/2}\Pi\cV\right)=\cV_\ell$. Thus,
\begin{align}
\left|\textnormal{SNR}_k(\rho)\right|^2=&\min_{\cV_\ell}\max_{\vect u\in\cV_\ell\setminus0}\fr{\vect u^TR_\rho\vect u}{\vect u^T\vect u}\nonumber\\
=&\lambda_k\left[R_\rho\right].
\end{align}
\end{proof}
Note that this interpretation as the eigenvalues of some operator fails for the mode SNR's since the latter's definition lacks the symplectic symmetry enjoyed by the definition of the mode temperatures.

\section{Details of illustrative examples}\label{appex}
Fig.~\ref{figex1} illustrated the application of our results to various special cases of state transformations under BLTO. Each subplot in the figure is associated with a particular initial state. The ones in the top half have single\hyp mode initial states, which can consequently be represented on the plot itself. However, this is not possible for the bottom subfigures, whose initial states are on two modes. For completeness, we provide below the details of all six examples used in the figure:
\begin{itemize}
\item Top left: Single\hyp mode state with principal variances $(3,1)$.
\item Top centre: Single\hyp mode state with principal variances $\left(3,4/3\right)$.
\item Top right: Single\hyp mode state with principal variances $(3,2.5)$.
\item Bottom left: Two\hyp mode state with covariance matrix $\left(\begin{array}{cccc}4&0&3.7&0\\0&4&0&-3.7\\3.7&0&4&0\\0&-3.7&0&4\end{array}\right)$.
\item Bottom centre: Two\hyp mode state with covariance matrix $\left(\begin{array}{cccc}4&0&1.6&0\\0&4&0&-1.6\\1.6&0&4&0\\0&-1.6&0&4\end{array}\right)$.
\item Bottom right: Two\hyp mode state with covariance matrix $\left(\begin{array}{cccc}4&0&1.73&0\\0&4&0&-1.73\\1.73&0&2.4&0\\0&-1.73&0&4\end{array}\right)$.
\end{itemize}

\end{document}